\documentclass[%
 reprint,
 amsmath,amssymb,
 aps,
]{revtex4-2}

\usepackage{graphicx}
\usepackage{dcolumn}
\usepackage{bm}
\usepackage{amsthm}

\newtheorem{lemma}{Lemma}

\newtheorem{theorem}{Theorem}

\newtheorem{corollary}{Corollary}

\newcommand{\bra}[1]{\langle {#1} |}
\newcommand{\ket}[1]{| {#1} \rangle}
\newcommand{\braket}[2]{\langle {#1} |{#2} \rangle}
\newcommand{\ketbra}[1]{| {#1} \rangle\langle {#1} |}

\newcommand{\lpnorm}[2]{\left\|#2\right\|_{#1}}
\newcommand{\trdist}[1]{\left\|#1\right\|_{\text{tr}}}
\newcommand{\fidelity}[2]{F\left(#1,#2\right)}
\newcommand{\tr}[1]{\text{tr}\left[#1\right]}
\newcommand{\ptr}[2]{\text{tr}_{#1}\left[#2\right]}
\newcommand{\diamondnorm}[1]{\left\|#1\right\|_{\diamond}}

\newcommand{\linop}[1]{\mathbf{L}\left(#1\right)}
\newcommand{\pos}[1]{\mathbf{Pos}\left(#1\right)}
\newcommand{\unitary}[1]{\mathbf{U}\left(#1\right)}
\newcommand{\dop}[1]{\mathbf{S}\left(#1\right)}
\newcommand{\puredop}[1]{\mathbf{P}\left(#1\right)}

\newcommand{\hh}{\mathcal{H}}
\newcommand{\vv}{\mathcal{V}}
\newcommand{\cd}{\mathbb{C}^d}
\newcommand{\cdim}[1]{\mathbb{C}^{#1}}
\newcommand{\rdim}[1]{\mathbb{R}^{#1}}
\newcommand{\vspan}[1]{\text{span}\left(#1\right)}
\newcommand{\nn}{\mathbb{N}}

\newcommand{\phiI}[1]{\phi_{#1}}
\newcommand{\psiI}[1]{\psi_{#1}}

\newcommand{\indicator}[1]{\text{I}\left[#1\right]}
\newcommand{\extpoint}[1]{\text{ext}\left(#1\right)}

\newcommand{\idop}{\mathbb{I}}

\newcommand{\eball}[2]{B_{#1}\left(#2\right)}

\begin{document}

\preprint{APS/123-QED}

\title{Quadratic improvement on accuracy of approximating pure quantum states and unitary gates by probabilistic implementation}

\author{Seiseki Akibue}
 \email{seiseki.akibue.rb@hco.ntt.co.jp}
\author{Go Kato}%
 \email{go.kato.gm@hco.ntt.co.jp}
\author{Seiichiro Tani}%
 \email{seiichiro.tani.cs@hco.ntt.co.jp}
\affiliation{%
NTT Communication Science Labs., NTT Corporation.\\
3--1, Morinosato-Wakamiya, Atsugi, Kanagawa 243-0198, Japan
}%

\date{\today}

\begin{abstract}
Pure quantum states are often approximately encoded as classical bit strings such as those representing probability amplitudes and those describing circuits that generate the quantum states.
The crucial quantity is the minimum length of classical bit strings from which the original pure states are approximately reconstructible. 
We derive asymptotically tight bounds on the minimum bit length required for probabilistic encodings with which one can approximately reconstruct the original pure state as an ensemble of the quantum states encoded in classical strings.
We also show that such a probabilistic encoding asymptotically halves the bit length required for ``deterministic" ones.
This is based on the fact that the accuracy of approximating pure states by using a given subset of pure states can be increased quadratically if we use ensembles of pure states in the subset. 
Moreover, we show that a similar fact holds when we consider the approximation of unitary gates by using a given subset of unitary gates. 
This improves the reduction rate of the circuit size by using probabilistic circuit synthesis compared to previous results.
This also demonstrates that the reduction is possible even for low-accuracy circuit synthesis, which might improve the accuracy of various NISQ algorithms.
\end{abstract}

\maketitle


\section{Introduction}
Pure quantum states are often approximately encoded in classical states in various quantum information processing tasks, such as classical bit strings storing probability amplitudes of pure states in the classical simulation of quantum circuits; classical data obtained by measurement in state tomography or state estimation; classical descriptions of quantum circuits that generate target pure states in quantum circuit synthesis \cite{VSI06, FWNRJ21}.
Recently, compact classical encodings from which one can predict probability distributions on outcomes obtained by measurements in certain classes have been developed \cite{HKP20,R99,GS19}.

The key issue is the minimum encoding. Here, we investigate it in terms of the minimum length $n$ of classical bit strings, with which one can approximately achieve any information processing task achievable with the original pure state on a $d$-dimensional system. That is, from the classical strings, one can construct a quantum state $\hat{\rho}$ that is indistinguishable from the original state $\phi$ within a certain accuracy by any measurements. In addition to {\it deterministic encodings}, which deterministically associate $\phi$ with a classical bit string, we consider {\it probabilistic encodings}, which associate $\phi$ with one of multiple classical strings according to some distribution. That is, in probabilistic encodings, $\phi$ is associated with an ensemble of classical strings from which $\hat{\rho}$ is constructed as an ensemble of the quantum states encoded in the classical strings.

Besides providing fundamental limits for information processing tasks using classical encodings, the minimum length $n$ is a fundamental quantity in various theoretical subjects including communication complexity, computational complexity, and asymptotic geometric analysis. Indeed, in deterministic encodings, classical strings do nothing but encode elements in an $\epsilon$-{\it covering} (sometimes called an $\epsilon$-{\it net}) of the set of pure states. Thus, the minimum bit length $n$ is the logarithm of the minimum size of $\epsilon$-coverings (called the {\it covering number}). Due to its prominent role in algorithm design and asymptotic geometric analysis, the covering number has been well-studied \cite{HLW06, GT13, ABmeetBanach}, and it is known that $n=O(d)$ bits are enough to encode an $\epsilon$-covering.
On the other hand, a particular task in communication complexity called {\it distributed quantum sampling}, which aims to classically transmit a pure state so as to approximately sample outcomes of arbitrary quantum measurement, provides a lower bound on the minimum length required for probabilistic encodings as $n=\Omega(d)$ \cite{M19}. Taking the two known facts into account, it seems that the minimum probabilistic encoding can be realized by a deterministic one.
This intuition is supported by the fact that an ensemble of deterministic classical states (called a {\it probabilistic classical state}) represents our lack of knowledge about a classical system while a pure state represents our maximum knowledge about a quantum system.

Contrary to this intuition, we show that the minimum length $n$ required for probabilistic encodings is exactly half of the one required for deterministic encodings in the asymptotic limits of the dimension or accuracy. Thus, the minimum encoding must associate some pure states with ensembles of classical strings describing distinct quantum states, which may be counter-intuitive, considering that pure states themselves are not probabilistic mixtures of distinct quantum states. The excessive bits required for deterministic encodings can be interpreted as a consequence of their excessive predictive capability such that they can not only reconstruct quantum states within a certain accuracy but also {\it deterministically} compute the probability distribution of any measurements within the same accuracy. Such a deterministic computation is impossible by using either the minimum (probabilistic) encoding or the original quantum states.

The bit length reduction by using probabilistic encodings follows from our refined estimation of the covering number and the following fact we prove: for any finite set $\mathcal{A}$ of pure states, it is possible to quadratically increase the accuracy of approximating arbitrary pure states by using {\it ensembles} of pure states in $\mathcal{A}$.
Moreover, we show that a similar fact holds when we consider the approximation of unitary gates by using a finite set of unitary gates. Recently, it has been found that when we approximately implement arbitrary unitary gates by using a gate sequence over a finite universal gate set (called {\it circuit synthesis}), the length of the gate sequence or the number of $T$ gates can be reduced by using ensembles of gate sequences for high-accuracy circuit synthesis \cite{C17, H17}. Our result improves the reduction rate of the previous results and shows that the reduction is possible even for low-accuracy circuit synthesis, which might improve the accuracy of NISQ algorithms.

\section{Preliminaries}
\label{sec:preliminaries}
In this section, we summarize basic notations used throughout the paper. Note that we consider only finite-dimensional Hilbert spaces. In particular, two-dimensional Hilbert space $\cdim{2}$ is called a qubit.
$\linop{\hh}$ and $\pos{\hh}$ represent the set of linear operators and positive semidefinite operators on Hilbert space $\hh$, respectively. $\dop{\hh}:=\left\{\rho\in\pos{\hh}:\tr{\rho}=1\right\}$ and $\puredop{\hh}:=\left\{\rho\in\dop{\hh}:\tr{\rho^2}=1\right\}$ represent the set of quantum states and that of pure states, respectively. Pure state $\phi\in\puredop{\hh}$ is sometimes alternatively represented by complex unit vector $\ket{\phi}\in\hh$ satisfying $\phi=\ketbra{\phi}$. Any physical transformation of the quantum state can be represented by a completely positive and trace preserving (CPTP) linear mapping $\Gamma:\linop{\hh}\rightarrow\linop{\hh'}$.

The trace distance $\trdist{\rho-\sigma}$ of two quantum states $\rho,\sigma\in\dop{\hh}$ is defined as $\trdist{M}:=\frac{1}{2}\tr{\sqrt{MM^\dag}}$ for $M\in\linop{\hh}$. It represents the maximum total variation distance between probability distributions obtained by measurements performed on two quantum states.
A similar notion measuring the distinguishability of $\rho$ and $\sigma$ is the fidelity function, defined by $\fidelity{\rho}{\sigma}:=\max\tr{\Phi^\rho\Phi^\sigma}$, where $\Phi^\rho\in\puredop{\hh\otimes\hh'}$ is a purification of $\rho$, i.e., $\rho=\ptr{\hh'}{\Phi^\rho}$, and the maximization is taken over all the purifications. Fuchs-van de Graaf inequalities \cite{FG99} provide relationships between the two measures with respect to the distinguishability as follows:
\begin{equation}
\label{ineq:FG}
 1-\sqrt{\fidelity{\rho}{\sigma}}\leq\trdist{\rho-\sigma}\leq\sqrt{1-\fidelity{\rho}{\sigma}}
\end{equation}
holds for any states $\rho,\sigma\in\dop{\hh}$, where the equality of the right inequality holds when $\rho$ and $\sigma$ are pure.

\section{Classical encoding of pure states}

\begin{figure}[h]
\includegraphics[height=.045\textheight]{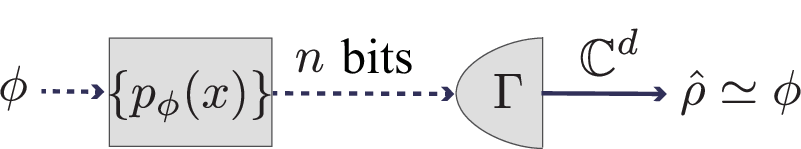}
\caption{\label{fig:setting} Probabilistic encoding of pure state $\phi$ on a $d$-dimensional system using $n$-bit strings and physical transformation $\Gamma$ corresponding to a decoder of the bit strings to quantum states. State $\phi$ is randomly encoded in label $x$ in finite set $X$ according to probability distribution $p_\phi:X\rightarrow[0,1]$, where $n=\lceil\log_2|X|\rceil$.}
\end{figure}

The existence of a probabilistic encoding is equivalent to the existence of physical transformation $\Gamma$ that can approximately reconstruct arbitrary pure state $\phi$ from probabilistic classical state $\{(p_\phi(x),x)\}_{x\in X}$ as shown in Fig.~\ref{fig:setting}. Physical transformation $\Gamma$ can be regarded as a {\it decoder} of the classical states to quantum states, which outputs mixed state $\rho_x\in\dop{\cd}$ when label $x\in X$ is inputted. Formally, $\Gamma$ is represented by a {\it classical-quantum channel} \cite{HSR03}, which is defined as $\Gamma(\sigma)=\sum_{x\in X}\bra{x}\sigma\ket{x}\rho_x$, where $\left\{\ket{x}\in\cdim{|X|}\right\}_{x\in X}$ is an orthonormal basis. We require that with the output state $\hat{\rho}:=\sum_x p_\phi(x)\rho_x$ of $\Gamma$, one can approximately sample any measurement outcomes performed on $\phi$ within total variation distance $\epsilon$, i.e., $ \trdist{\phi-\hat{\rho}}<\epsilon$ for given $\epsilon\in(0,1]$. Thus, we say that a probabilistic encoding of $\puredop{\cd}$ with accuracy $\epsilon$ exists if and only if there exists set $\{\rho_x\in\dop{\cd}\}_{x\in X}$ of quantum states satisfying
\begin{equation}
\label{eq:reconstruction}
\max_{\phi\in\puredop{\cd}}\min_{p}\trdist{\phi-\sum_{x\in X}p(x)\rho_x}<\epsilon,
\end{equation}
where the minimum is taken over probability distribution $p$ over $X$, i.e, $\sum_{x\in X}p(x)=1$ and $p(x)\geq0$.
Note that since any mixed state is a probabilistic mixture of pure states and trace distance is convex, Eq.~\eqref{eq:reconstruction} also guarantees that an arbitrary mixed state is also approximately reconstructible within the same accuracy as pure states.

\subsection{Minimum deterministic encoding}
First, we consider deterministic encodings, which can be defined as particular probabilistic encodings. Concretely, every pure state $\phi$ is encoded into a single label $x_\phi\in X$, which implies the output state of $\Gamma$ is $\hat{\rho}=\rho_{x_\phi}$.
Thus, we say that a deterministic encoding of $\puredop{\cd}$ with accuracy $\epsilon$ exists if and only if there exists set $\{\rho_x\in\dop{\cd}\}_{x\in X}$ of quantum states satisfying
\begin{equation}
\label{eq:covering}
\max_{\phi\in\puredop{\cd}}\min_{x\in X}\trdist{\phi-\rho_x}<\epsilon,
\end{equation}
which is called an {\it external} $\epsilon$-covering of $\puredop{\cd}$. A set of pure states $\{\rho_x\in\puredop{\cd}\}_{x\in X}$ satisfying Eq.~\eqref{eq:covering} is called an {\it internal} $\epsilon$-covering of $\puredop{\cd}$, which corresponds to particular deterministic encodings such as one storing probability amplitudes approximately representing $\phi$. The minimum size of internal (or external) $\epsilon$-coverings is called the internal (or external) covering number and denoted by $I_{in}$ (or $I_{ex}$.) Note that $I_{ex}\leq I_{in}$ by definition and the minimum bit length $n$ required for deterministic encodings equals to $\lceil\log_2 I_{ex}\rceil$.

Since the condition of the $\epsilon$-coverings in Eq.~\eqref{eq:covering} is equivalent to that for the set of $\epsilon$-balls $\left\{\eball{\epsilon}{\rho_x}:=\left\{\psi\in\puredop{\cd}:\trdist{\psi-\rho_x}<\epsilon\right\}\right\}_x$ to cover $\puredop{\cd}$, a detailed analysis of the volume of the $\epsilon$-ball provides a good estimation of the covering numbers. As shown in Appendix \ref{appendix:incovering}, the volume can be calculated as $\mu(\eball{\epsilon}{\phi})=\epsilon^{2(d-1)}$ with respect to the unitarily invariant probability measure $\mu$ for any $\phi\in\puredop{\cd}$. This directly provides a lower bound on $I_{in}$ and also its upper bound by applying the method of constructing an internal $\epsilon$-covering developed in \cite{ABmeetBanach}.
We obtain the following estimation of $I_{in}$, which is tighter than previous estimations \cite{HLW06, GT13} in large dimensions. For completeness, we provide a construction of an internal $\epsilon$-covering and an estimation of the parameters appearing in the construction in Appendix \ref{appendix:construction}.
\begin{lemma}
\label{lemma:internalcovering}
For any $\epsilon\in(0,1]$ and a positive integer $d\in\nn$ specified below, the internal covering number $I_{in}$ of internal $\epsilon$-coverings of $\puredop{\cd}$ is bounded as follows:  For any $r>2$, there exists $d_0\in\nn$ such that
\begin{equation}
\left(\frac{1}{\epsilon}\right)^{2(d-1)}\leq I_{in}\leq rd(\ln d)\left(\frac{1}{\epsilon}\right)^{2(d-1)},
\end{equation}
where the left inequality holds for any $d\in\nn$, and the right inequality holds for any $d\geq d_0$. 
For example, if $r=5$, we can set $d_0=2$.
\end{lemma}

To obtain a lower bound on $I_{ex}$, we use the following upper bounds on the volume of the $\epsilon$-ball as shown in Appendix \ref{appendix:excovering}: 
\begin{equation}
\label{eq:upperboundofextintersection}
\forall \epsilon\in\left(0,\frac{1}{2}\right], \mu(\eball{\epsilon}{\rho})\leq
 \left\{
 \begin{array}{l}
 (2\epsilon)^{2(d-1)}\ \ \text{for}\ 3\geq d\geq1\\
 \epsilon^{2(d-1)}\ \ \ \ \ \text{for}\ d\geq4.
\end{array}
\right.
\end{equation}
This bound and $\mu(\eball{\epsilon}{\phi})=\epsilon^{2(d-1)}$ imply that the volume of the $\epsilon$-ball can be maximized by setting its center as a pure state if $d\geq4$, which is contrary to what happens in a qubit ($d=2$), where $\eball{\epsilon}{\rho}$ corresponds to the intersection of the Bloch sphere and a ball centered at $\rho$ and the intersection is maximized not by a ball centered at a point {\it on} the Bloch sphere but by a ball centered at a point {\it inside} the Bloch ball. The qubit case also implies that the condition $d\geq4$ for the second inequality cannot be fully relaxed. $\mu(\eball{\epsilon}{\sigma})=1$ if $\epsilon>1-\frac{1}{d}$ with the maximally mixed state $\sigma=\frac{1}{d}\idop$ implies another condition $\epsilon\in\left(0,\frac{1}{2}\right]$ is also not fully removable. By using Eq.~\eqref{eq:upperboundofextintersection}, we easily obtain the following lower bound on $I_{ex}$.

\begin{lemma}
\label{lemma:externalcovering}
For any $\epsilon\in\left(0,\frac{1}{2}\right]$ and a positive integer $d\in\nn$ specified below, the external covering number $I_{ex}$ of external $\epsilon$-coverings of $\puredop{\cd}$ is bounded as follows:
\begin{equation}
 I_{ex}\geq
  \left\{
 \begin{array}{l}
\left(\frac{1}{2\epsilon}\right)^{2(d-1)}\ \ \text{for}\ 3\geq d\geq1 \\
\left(\frac{1}{\epsilon}\right)^{2(d-1)}\ \ \ \text{for}\ d\geq4.
\end{array}
\right.
\end{equation}
\end{lemma}

Using the two lemmas by setting $r=5$, we obtain the following theorem straightforwardly.
\begin{theorem}
\label{thm:definiteencoding}
For any $\epsilon\in\left(0,\frac{1}{2}\right]$ and an integer $d\geq2$ specified below,
the minimum size of label set $X$ used over all deterministic encodings of $\puredop{\cd}$ with accuracy $\epsilon$ is bounded by 
\begin{equation}
2\cdot l(d,2\epsilon)\leq\log_2 |X|\leq 2\cdot l(d,\epsilon)+\log_2(5d\ln d),
\end{equation}
where $l(d,\epsilon):=\left(d-1\right)\log_2\left(\frac{1}{\epsilon}\right)$. Moreover, if $d\geq4$, the lower bound can be strengthened as $2\cdot l(d,\epsilon)\leq\log_2 |X|$.
\end{theorem}

Using Theorem \ref{thm:definiteencoding} and $n=\lceil\log_2 |X|\rceil$, we obtain the asymptotic bit rate per dimension $\lim_{d\rightarrow \infty}\frac{n}{d}=2\log_2\left(\frac{1}{\epsilon}\right)$ of the minimum deterministic encoding for fixed $\epsilon\in\left(0,\frac{1}{2}\right]$. 
We can also obtain the asymptotic bit rate per accuracy $\lim_{\epsilon\rightarrow 0}\frac{n}{-\log_2\epsilon}=2(d-1)$ of the minimum deterministic encoding for fixed $d\geq2$.

\subsection{Minimum probabilistic encoding}
We prove the existence of a probabilistic encoding that achieves exactly half the asymptotic bit length required for the minimum deterministic encoding, and its optimality. The main tool for the proof is the following minimax relationship between the fidelity and the trace distance.
\begin{lemma}
\label{lemma:minimax}
 For any CPTP linear mapping $\Lambda:\linop{\hh'}\rightarrow\linop{\hh}$, it holds that
 \begin{eqnarray}
 \label{eq:minimax}
&&\max_{\phi\in\puredop{\hh}}\min_{\sigma\in\dop{\hh'}}\trdist{\phi-\Lambda(\sigma)}\nonumber\\
&&=1-\min_{\phi\in\puredop{\hh}}\max_{\psi\in\puredop{\hh'}} \fidelity{\Lambda(\psi)}{\phi}.
\end{eqnarray}
\end{lemma}
\begin{proof}
We use the minimax theorem as follows:
 \begin{eqnarray}
 &&(L.H.S.)\nonumber\\
 &&=\max_{\phi\in\puredop{\hh}}\min_{\sigma\in\dop{\hh'}}\max_{0\leq M\leq\idop}\tr{M(\phi-\Lambda(\sigma))}\nonumber\\
 &&=\max_{\phi\in\puredop{\hh}}\max_{0\leq M\leq\idop}\min_{\sigma\in\dop{\hh'}}\tr{M(\phi-\Lambda(\sigma))}\nonumber\\
 &&=\max_{0\leq M\leq\idop}\left(\max_{\phi\in\puredop{\hh}}\tr{M\phi}-\max_{\psi\in\puredop{\hh'}}\tr{M\Lambda(\psi)}\right)\nonumber\\
 &&=(R.H.S.)
\end{eqnarray}
Note that the minimax theorem, used in the second equation, is applicable since $f(\sigma,M):=\tr{M(\phi-\Lambda(\sigma))}$ is affine with respect to each variable and the domain of $M$ and $\sigma$ are compact and convex. The last equality holds since the maximum is achieved if $\text{rank}M=1$.
\end{proof}

As a special case of Lemma \ref{lemma:minimax}, we obtain the following lemma about the relationship of the accuracy of approximating pure states by a finite number of pure states and that by their ensembles:
\begin{corollary}
\label{corollary:statemixing}
 Let $\{\phi_x\in\puredop{\hh}\}_{x\in X}$ be a finite set of pure states. Then, it holds that
\begin{equation}
\label{eq:statemixing}
 \max_{\phi\in\puredop{\hh}}\min_{p}\trdist{\phi-\sum_{x\in X} p(x)\phi_x}= \max_{\phi\in\puredop{\hh}}\min_{x\in X}\trdist{\phi-\phi_x}^2,
\end{equation}
where the first minimum is taken over probability distribution $p$ over $X$.
\end{corollary}
\begin{proof}
Suppose $\Lambda$ in Lemma \ref{lemma:minimax} is a classical-quantum channel such that $\Lambda(\sigma)=\sum_{x\in X}\bra{x}\sigma\ket{x}\phi_x$, which corresponds to a particular decoder for a classical encoding that outputs pure state $\phi_x\in\puredop{\hh}$ when label $x$ is inputted. Then, L.H.S. of Eq.~\eqref{eq:minimax} is equal to that of Eq.~\eqref{eq:statemixing} by interchanging $\bra{x}\sigma\ket{x}$ and $p(x)$. and R.H.S. of  Eq.~\eqref{eq:minimax} is equal to
\begin{eqnarray}
 &&1-\min_{\phi\in\puredop{\hh}}\max_{p} \fidelity{\sum_{x\in X}p(x)\phi_x}{\phi}\nonumber\\
 &&=1-\min_{\phi\in\puredop{\hh}}\max_{p} \sum_{x\in X}p(x)\fidelity{\phi_x}{\phi}\nonumber\\
 &&=1-\min_{\phi\in\puredop{\hh}}\max_{x\in X} \fidelity{\phi_x}{\phi},
\end{eqnarray}
which is equal to R.H.S. of Eq.~\eqref{eq:statemixing} since the equality of the second inequlity in Eq.~\eqref{ineq:FG} holds when $\rho$ and $\sigma$ are pure.
\end{proof}
This corollary implies that an internal $\sqrt{\epsilon}$-covering is sufficient to approximate arbitrary pure state within accuracy $\epsilon$ by using its probabilistic mixture. This can be intuitively understood by the curvature of the sphere as illustrated in Fig.~\ref{fig:Bloch}. Indeed, Corollary \ref{corollary:statemixing} (for $\hh=\cdim{2}$) and the Bloch representation imply that for any compact and convex set $K$ whose extreme points $\extpoint{K}$ reside on sphere $S$ with radius $\frac{1}{2}$, $\delta=\sqrt{\epsilon}$ holds, where $\epsilon$ and $\delta$ are the distance between $K$ and the farthest point on $S$ from $K$ and that between $\extpoint{K}$ and the farthest point on $S$ from $\extpoint{K}$, respectively. This can also be derived from elementary geometric observations.

\begin{figure}[h]
\includegraphics[height=.18\textheight]{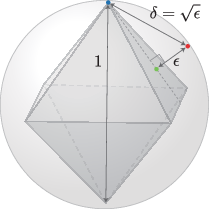}
\caption{\label{fig:Bloch} For any pure qubit state $\phi$, we can find probability mixture $\hat{\rho}$ of six pure states, which are the eigenstates of the Pauli operators and represented by the extreme points of the octahedron on the Bloch sphere, such that $\trdist{\phi-\hat{\rho}}\leq\epsilon=\frac{1}{2\sqrt{3}}\left(\sqrt{3}-1\right)$. If we represent the Bloch sphere by a sphere with radius $\frac{1}{2}$, $\epsilon$ is the longest Euclidean distance between a point on the Bloch sphere and the octahedron since the trace distance between two quantum states is equal to the Euclidean distance between the corresponding points in the Bloch ball. On the other hand, the longest trace distance $\delta$ between a pure state and the six pure states, which is equivalent to the longest Euclidean distance between a point on the Bloch sphere and the extreme points of the octahedron, satisfies $\delta=\sqrt{\epsilon}$.}
\end{figure}

By using Lemma \ref{lemma:minimax} and Corollary \ref{corollary:statemixing}, we can derive following asymptotically tight bounds on the minimum bit length required for probabilistic encodings:

\begin{theorem}
\label{thm:classicalencoding}
For any $\epsilon\in\left(0,1\right]$ and an integer $d\geq2$,
the minimum size of label set $X$ used over all probabilistic encodings of $\puredop{\cd}$ with accuracy $\epsilon$ is bounded by 
\begin{equation}
l(d,\epsilon)-\log_2d\leq\log_2|X|\leq l(d,\epsilon)+\log_2(5d\ln d),
\end{equation}
where $l(d,\epsilon):=\left(d-1\right)\log_2\left(\frac{1}{\epsilon}\right)$.
\end{theorem}
\begin{proof}
When $\{\phi_x\in\puredop{\cd}\}_{x\in X}$ is an internal $\sqrt{\epsilon}$-covering, Corollary \ref{corollary:statemixing} implies $\{\phi_x\}_{x\in X}$ satisfies Eq.~\eqref{eq:reconstruction}. Thus, there exists a probabilistic encoding of $\puredop{\cd}$ with accuracy $\epsilon$ and label set $X$, whose size is upper bounded by using Lemma \ref{lemma:internalcovering} with setting $r=5$.

Next, we show the lower bound. Let $\{\rho_x\in\dop{\cd}\}_{x\in X}$ satisfy Eq.~\eqref{eq:reconstruction}. By using Lemma \ref{lemma:minimax} and setting $\Lambda$ a classical-quantum channel such that $\Lambda(\sigma)=\sum_{x\in X}\bra{x}\sigma\ket{x}\rho_x$ as in the proof of Corollary \ref{corollary:statemixing}, we obtain
\begin{eqnarray}
1-\epsilon&<&\min_{\phi\in\puredop{\hh}}\max_{p}\sum_{x\in X}p(x) \fidelity{\rho_x}{\phi}\nonumber\\
&=&\min_{\phi\in\puredop{\hh}}\max_{x\in X} \fidelity{\rho_x}{\phi}.
\end{eqnarray}

By letting $\rho_x=\sum_{i=1}^{d}p(i|x)\phiI{i|x}$, we obtain that for any $\phi\in\puredop{\cd}$, there exists $i$ and $x$ such that
\begin{eqnarray}
 1-\epsilon&<&\fidelity{\rho_x}{\phi}=\sum_{j=1}^{d}p(j|x)\fidelity{\phiI{j|x}}{\phi}\nonumber\\
 &\leq&\fidelity{\phiI{i|x}}{\phi}=1-\trdist{\phi-\phiI{i|x}}^2.
\end{eqnarray}
Thus, $\{\phiI{i|x}\}_{i,x}$ is an internal $\sqrt{\epsilon}$-covering of $\puredop{\cd}$. Hence, the lower bound can be obtained by applying Lemma \ref{lemma:internalcovering} as $|X|d\geq\left(\frac{1}{\sqrt{\epsilon}}\right)^{2(d-1)}$.
\end{proof}

Using Theorem \ref{thm:classicalencoding} and $n=\lceil\log_2 |X|\rceil$, we obtain the asymptotic bit rate per dimension $\lim_{d\rightarrow \infty}\frac{n}{d}=\log_2\left(\frac{1}{\epsilon}\right)$ of the minimum probabilistic encoding for fixed $\epsilon\in\left(0,1\right]$, and one per accuracy $\lim_{\epsilon\rightarrow 0}\frac{n}{-\log_2\epsilon}=d-1$ of the minimum probabilistic encoding for fixed $d\geq2$, which are exactly half of those of the minimum deterministic encoding.

\section{Probabilistic circuit synthesis}
In the context of circuit synthesis using a finite universal gate set, it has recently been found that the length of the gate sequence or the number of $T$ gates can be reduced by using ensembles of gate sequences \cite{C17, H17}. This is based on the so-called {\it mixing lemma}, which shows that if finite set $\{\Upsilon_x\}_{x}$ of unitary transformations approximates arbitrary unitary transformations within sufficient high accuracy, it is possible to increase the accuracy by using {\it particular} ensembles $\sum_xp(x)\Upsilon_x$.

Based on Corollary \ref{corollary:statemixing}, where we derive the accuracy achieved by the {\it optimal} ensembles of pure states in approximating arbitrary pure states, we can derive bounds on the accuracy achieved by the {\it optimal} ensembles of unitary transformations in approximating arbitrary unitary transformations in the following theorem.
Note that similarly to \cite{C17, H17}, we measure the accuracy of the approximation by using the {\it diamond norm} (sometimes called the completely bounded trace norm) of hermitian preserving linear mappings, defined as
\begin{equation}
\frac{1}{2} \diamondnorm{\Xi}:=\max_{\Phi\in\puredop{\hh_1\otimes\hh_1'}}\trdist{\Xi\otimes id_{\hh_1'}(\Phi)},
\end{equation}
where $\Xi:\linop{\hh_1}\rightarrow\linop{\hh_2}$, $id_\hh$ is the identity mapping on $\linop{\hh}$ and $\dim\hh_1'=\dim\hh_1$.
\begin{theorem}
\label{thm:unitary}
For an integer $d\geq2$ specified below, let $\{\Upsilon_x\}_{x\in X}$ be a finite set of unitary transformations on $\linop{\cd}$. Then, it holds that
\begin{eqnarray}
\frac{2}{d}\alpha\leq \max_{\Upsilon}\min_{p}\frac{1}{2}\diamondnorm{\Upsilon-\sum_{x\in X}p(x)\Upsilon_x}\leq \alpha,\\
 \alpha= \max_{\Upsilon}\min_{x\in X}\left(\frac{1}{2}\diamondnorm{\Upsilon-\Upsilon_x}\right)^2,
 \end{eqnarray}
where the first minimum is taken over probability distribution $p$ over $X$. 
Note that if $d=2$, the equalities hold.
\end{theorem}
 This theorem resembles Corollary \ref{corollary:statemixing} (especially when $d=2$), which can be regarded as a consequence of the similarity between pure states and unitary transformations via the Choi-Jamio\l kowski representation. Although a proof uses the minimax theorem as in the proof of Corollary \ref{corollary:statemixing}, it requires additional work to some extent. We give a complete proof in Appendix \ref{appendix:unitarymixing} with the sharp lower bound.

This theorem implies that quantum circuit synthesis using probabilistically generated circuits formed from finite universal gate set $\mathcal{C}=\{g_1,g_2,\cdots\}$ can reduce the circuit size. 
To see this, first, let $\{\Upsilon_x^{(n)}\}_x$ be the set of unitary transformations representing the unitary circuit realized by gate sequences $g_{i_1}\cdots g_{i_n}$ of length at most $n$. Next, let $n(\epsilon)$ be the smallest length of gate sequences to approximate arbitrary unitary transformations within accuracy $\epsilon$, i.e., $n(\epsilon):=\min\{n\in\nn:\max_{\Upsilon}\min_{x}\frac{1}{2}\diamondnorm{\Upsilon-\Upsilon_x^{(n)}}<\epsilon\}$.
Theorem \ref{thm:unitary} implies that by using the probabilistic implementation, we can implement arbitrary unitary transformation within accuracy $\epsilon$ only with an $n(\sqrt{\epsilon})$-size circuit.

The accuracy in approximating unitary transformations is often measured by using the operator norm $\lpnorm{\infty}{X}:=\max_{\phi\in\puredop{\hh}}\lpnorm{2}{X\ket{\phi}}$. The celebrated Solovay-Kitaev theorem shows that for any finite universal gate set $\mathcal{C}$, set $\{U_x^{(n)}\}_x$ of unitary circuits realized by gate sequences of length $n=O\left(\log^c\left(\frac{1}{\delta}\right)\right)$ ($c\geq1$) is sufficient for approximating arbitrary unitary operators, i.e., $\max_U\min_x\lpnorm{\infty}{U-U_x^{(n)}}<\delta$, where we denote a unitary circuit and the unitary operator representing the circuit by the same symbol $U_x^{(n)}$. On the other hand, a relationship between the operator norm and the diamond norm shown in Appendix \ref{appendix:diamondnorm} implies that $\diamondnorm{\Upsilon-\Upsilon_x}<\delta\sqrt{4-\delta^2}$ if $\lpnorm{\infty}{U-U_x}<\delta$, where $\Upsilon(\rho)=U\rho U^\dag$ and $\Upsilon_x(\rho)=U_x\rho U_x^\dag$. Combining with Theorem \ref{thm:unitary}, we can verify that arbitrary unitary transformations can be approximated by ensembles of $\{\Upsilon_x\}_x$ such that
\begin{equation}
\max_{\Upsilon}\min_{p}\frac{1}{2}\diamondnorm{\Upsilon-\sum_{x}p(x)\Upsilon_x}<\delta^2\left(1-\left(\frac{\delta}{2}\right)^2\right)
\end{equation}
if $\max_U\min_x\lpnorm{\infty}{U-U_x}<\delta$. This bound is tighter than the previous bound obtained in \cite[Theorem~1]{C17}, $\max_{\Upsilon}\min_{p}\frac{1}{2}\diamondnorm{\Upsilon-\sum_{x}p(x)\Upsilon_x}<5\delta^2$. Moreover, our bound holds if $\delta\leq\sqrt{2}$ while the previous bound was shown to hold for $\delta<0.01$. In addition to the improved estimation, our lower bound reveals the limitation of the probabilistic implemetation.

As suggested by the Solovay-Kitaev theorem, if we assume $n(\epsilon)\sim a\log^c\left(\frac{1}{\epsilon}\right)$ in the high-accuracy regime ($\epsilon\ll 1$), the probabilistic implementation reduces the length of gate sequences by about $1-\left(\frac{1}{2}\right)^c\geq50\%$ since $c\geq1$ from a volume consideration. Even in the low-accuracy regime, our bound guarantees the reduction. For example, we show how the gate length can be reduced by using the probabilistic implementation to synthesize single-qubit unitary transformations with gate set $\{S,H,T\}$ in Appendix \ref{appendix:synthesis}.

\section{Conclusion}
In this paper, we have considered the minimum probabilistic encoding so as to approximately reconstruct an arbitrary pure state $\phi\in\puredop{\cd}$ from an $n$-bit string within accuracy $\epsilon$ with respect to the trace distance. We then demonstrated that it cannot be realized by simply storing an element of the minimum $\epsilon$-covering. More precisely, we proved that the bit rate required for probabilistic encodings is exactly half of that of the minimum length of bits necessary to store elements of an $\epsilon$-covering of $\puredop{\cd}$ in asymptotic limit $\epsilon\rightarrow0$ or in limit $d\rightarrow\infty$ when $\epsilon\in\left(0,\frac{1}{2}\right]$. In limit $d\rightarrow\infty$ when $\epsilon\in\left(\frac{1}{2},1\right]$, the same result holds if we consider only internal $\epsilon$-coverings; however, in general, whether the same result holds or not is an open problem. Several numerical calculations suggest a positive answer.

Moreover, we show that similarly to the state encoding, for any finite set $\{\Upsilon_x\}_x$ of unitary transformations, we can at least quadratically increase the accuracy of approximating arbitrary unitary transformations by using ensembles of $\{\Upsilon_x\}_x$. Particularly, we obtain bounds on the accuracy when one uses the optimal ensembles to reveal the possibility and the limitation of the probabilistic implementation of unitary transformations.

Our result could provide a new quantitative guiding principle to explore further capabilities and limitations of manipulating a quantum system as well as the foundations of quantum theory, including the following two related topics:

\begin{enumerate}
 \item The results demonstrate an information-theoretical separation of the memory size to store a pure quantum state between strong simulations and weak ones, which are two types of classical simulation of a quantum computer \cite{N10, BBCCGH19, HMW20} (the former approximately {\it computes} the probability distribution over the outcomes, whereas the latter only approximately {\it samples} the outcomes.) 
  
 \item The complex projective space representation of pure states can be regarded as a classical encoding of pure states in deterministic classical states describing operators in $\puredop{\hh}$. The fact that any distinct pure states are encoded indistinguishable classical states inclines us to think that the indistinguishability of non-orthogonal pure states results from our limited ability to measure them.
To interpret the indistinguishability as an intrinsic feature of pure states, classical encodings of pure states in probabilistic classical states have been constructed in $\psi$-{\it epistemic} models \cite{LJBT12, ABCL13}, in which the encodings use indistinguishable and probabilistic classical states to encode some distinct pure states.
Our results show that indistinguishable classical states encoding distinct elements in an $\epsilon$-covering are not only helpful for such an interpretation but also necessary for the minimum probabilistic encoding.
\end{enumerate}

\begin{acknowledgments}
We thank Yuki Takeuchi, Yasunari Suzuki, Yasuhiro Takahashi, and Adel Sohbi for helpful discussions.
This work was partially supported by JST Moonshot R\&D MILLENNIA Program (Grant No.JPMJMS2061).
SA was partially supported by JST, PRESTO Grant No.JPMJPR2111 and JPMXS0120319794.
GK was supported in part by the Grant-in-Aid for Scientific Research (C) No.20K03779, (C) No.21K03388, and (S) No.18H05237 of JSPS, CREST (Japan Science and Technology Agency) Grant No.JPMJCR1671.
ST was partially supported by the Grant-in-Aid for Transformative Research Areas No.JP20H05966 of JSPS.
\end{acknowledgments}

\onecolumngrid

\bibliography{references.bib}

\newpage
\appendix

\section{Volume of $\epsilon$-ball in $\puredop{\cd}$}
\label{appendix:incovering}
To construct an $\epsilon$-covering, we first derive the volume of $\epsilon$-ball $\eball{\epsilon}{\phi}:=\left\{\psi\in\puredop{\cd}:\trdist{\psi-\phi}<\epsilon\right\}$ in $\puredop{\cd}$ as follows:
\begin{equation}
\label{eq:volume_of_eball}
\forall d\in\nn,\forall \epsilon\in(0,1],\forall \phi\in\puredop{\cd}, \mu(\eball{\epsilon}{\phi})=\epsilon^{2(d-1)},
\end{equation}
where $\mu$ is the unitarily invariant probability measure on the Borel sets of $\puredop{\cd}$.

When $d=1$, Eq.~\eqref{eq:volume_of_eball} holds. By assuming $d\geq2$, we proceed as follows:
\begin{eqnarray}
&&\mu(\eball{\epsilon}{\phi})\nonumber\\
&&=\mu\left(\left\{\psi\in\puredop{\cd}:\trdist{\ketbra{0}-\psi}<\epsilon\right\}\right)\nonumber\\
&&=\mu\left(\left\{\psi\in\puredop{\cd}:|\braket{0}{\psi}|^2>1-\epsilon^2\right\}\right)\nonumber\\
&&=\xi\left(\left\{\vec{x}\in\rdim{2d}:\lpnorm{2}{\vec{x}}=1\wedge x_1^2+x_2^2>1-\epsilon^2\right\}\right),
\end{eqnarray}
where the first equality uses fixed pure state $\ket{0}$ and the unitary invariance of $\mu$ and the trace distance, the second equality uses Eq.~\eqref{ineq:FG}, and the third equality uses the relationship between $\mu$ and the uniform spherical probability measure $\xi$. Using a spherical coordinate system, we can proceed as follows:
\begin{eqnarray}
\mu(\eball{\epsilon}{\phi})&=&\frac{V(\epsilon)}{V(1)},\\
\text{where}\ V(\epsilon)&:=&\int_{D_\epsilon}\sin^{2d-2}\theta\sin^{2d-3}\phi d\theta d\phi\nonumber\\
&=&4\int_{\hat{D}_\epsilon}\sin^{2d-2}\theta\sin^{2d-3}\phi d\theta d\phi
\end{eqnarray}
and the domain of the integration $D_\epsilon$ is given by $\{(\theta,\phi):\theta,\phi\in(0,\pi),\sin\theta\sin\phi<\epsilon\}$. Since the domain and that of the integrand have reflection symmetries about two lines $\theta=\frac{\pi}{2}$ and $\phi=\frac{\pi}{2}$, it is sufficient to perform the integration in domain $\hat{D}_\epsilon:=\{(\theta,\phi):\theta,\phi\in\left(0,\frac{\pi}{2}\right),\sin\theta\sin\phi<\epsilon\}$. By changing the variables as
$ \begin{pmatrix}
 x\\y
\end{pmatrix}
=
\begin{pmatrix}
 \sin\theta\sin\phi\\
 \sin\theta
\end{pmatrix}$, we obtain
\begin{eqnarray}
 V(\epsilon)&=&4\int_0^{\epsilon}dxx^{2d-3}\int_x^1dy\frac{y}{\sqrt{1-y^2}\sqrt{y^2-x^2}}\nonumber\\
 &=&4\int_0^{\epsilon}dxx^{2d-3}\left[\arcsin\sqrt{\frac{1-y^2}{1-x^2}}\right]_1^x\nonumber\\
 &=&\frac{\pi}{d-1}\epsilon^{2(d-1)}
\end{eqnarray}
for $\epsilon\in[0,1]$. This completes the calculation.

\section{Upper bound for internal covering number $I_{in}$}
\label{appendix:construction}
We construct an internal $\epsilon$-covering ($\epsilon\in(0,1]$) following the proof in \cite[Corollary 5.5]{ABmeetBanach}. The construction is based on the fact that sufficiently many pure states randomly sampled form an $\epsilon$-covering. However, since the probability of a new random pure state residing in the uncovered region decreases when many random $\epsilon$-balls are sampled, it is better to stop sampling a pure state and change the strategy of the construction.

In the proof, we represent some parameters explicitly, which are tailored to the $\epsilon$-covering with respect to the trace distance.
 Assume $d\geq2$ and let $D=2(d-1)(\geq2)$. Let $\{\phiI{j}\in\puredop{\cd}\}_{j=1}^{J_R}$ be a set of finite randomly sampled pure states with respect to product measure $\mu^{J_R}$. The expected volume of the region not covered by $A:=\cup_{j=1}^{J_R}\eball{\epsilon_R}{\phiI{j}}$ ($0<\epsilon_R\leq 1$) can be calculated as follows:
 \begin{eqnarray}
 &&\int d\mu^{J_R}\mu\left(A^c\right)\nonumber\\
 &&= \int d\mu^{J_R}\int d\mu(\psi)\prod_{j=1}^{J_R}\indicator{\trdist{\psi-\phiI{j}}\geq\epsilon_R}\nonumber\\
 &&=\int d\mu(\psi)\prod_{j=1}^{J_R} \int d\mu(\phiI{j})\indicator{\trdist{\psi-\phiI{j}}\geq\epsilon_R}\nonumber\\
 &&=\left(1-\epsilon_R^D\right)^{J_R}\leq\exp\left(-J_R\epsilon_R^D\right),
\end{eqnarray}
where we use Fubini's theorem and Eq.~\eqref{eq:volume_of_eball} in the second equation and the third equation, respectively. Note that $\indicator{X}\in\{0,1\}$ is the indicator function, i.e., $\indicator{X}=1$ iff $X$ is true.

Thus, there exists $\{\phiI{j}\}_{j=1}^{J_R}$ such that $\mu\left(A^c\right)\leq\exp\left(-J_R\epsilon_R^D\right)$. Pick $\{\psiI{j}\}_{j=1}^{J_P}$ as much as possible such that $\eball{\epsilon_P}{\psiI{j}}$ are disjoint and contained in $A^c$. When $0<\epsilon_P\leq\epsilon_R\leq 1$, we can verify that $\{\phiI{j}\}_{j=1}^{J_R}\cup\{\psiI{j}\}_{j=1}^{J_P}$ is an $(\epsilon_R+\epsilon_P)$-covering  and its size $J:=J_R+J_P$ is upper bounded as
\begin{equation}
 J\leq J_R+\frac{\exp\left(-J_R\epsilon_R^D\right)}{\epsilon_P^D}.
\end{equation}
By setting $J_R=\left\lceil\frac{D}{\epsilon_R^D}\ln\left(\frac{\epsilon_R}{\epsilon_P}\right)\right\rceil$, $\epsilon_P=\frac{\epsilon_R}{x}$ and $\epsilon_R=\frac{x}{1+x}\epsilon$ with $x\geq1$, we obtain the following upper bound:
\begin{equation}
 J\leq\left\lceil \frac{D\ln x}{\epsilon_R^D}\right\rceil+\frac{1}{\epsilon_R^D}
\leq\frac{1}{\epsilon^D}\left\{\left(1+\frac{1}{x}\right)^D(D\ln x+1)+1\right\}=\frac{2d\ln d}{\epsilon^D}\cdot\frac{\alpha(d,x)}{2d\ln d},
\end{equation}
where $\alpha(d,x)=\left(1+\frac{1}{x}\right)^D(D\ln x+1)+1$.
Since $\lim_{d\rightarrow\infty}\frac{\alpha(d,D\ln D)}{2d\ln d}=1$, we obtain that for any $r>2$ there exists $d_0\in\nn$ such that
\begin{equation}
\forall d\geq d_0,\forall\epsilon\in(0,1], J\leq rd(\ln d)\left(\frac{1}{\epsilon}\right)^{2(d-1)}.
\end{equation}
For example, if $r=5$, we can set $d_0=2$.
This completes the proof.

\section{Lower bound for external covering number $I_{ex}$}
\label{appendix:excovering}
We can derive a lower bound for the external covering number as a direct consequence of the following upper bounds on the volume of the maximum intersection of $\epsilon$-ball in $\dop{\cd}$ and $\puredop{\cd}$, which will be proven in this section: For any $\rho\in\dop{\cd}$ and any $\epsilon\in\left(0,\frac{1}{2}\right]$,
\begin{eqnarray}
 \label{ineq:volume_of_exeball1}
 \forall d\geq1,\mu(\eball{\epsilon}{\rho})&\leq&(2\epsilon)^{2(d-1)},\\
 \label{ineq:volume_of_exeball}
  \forall d\geq4,\mu(\eball{\epsilon}{\rho})&\leq&\epsilon^{2(d-1)}.
\end{eqnarray}
Combined with Eq.~\eqref{eq:volume_of_eball}, Eq.~\eqref{ineq:volume_of_exeball} implies that the maximum intersection is achieved when $\rho$ is pure if two conditions $d\geq4$ and $\epsilon\in\left(0,\frac{1}{2}\right]$ are satisfied. These two conditions are not tight but cannot be fully relaxed since $\mu(\eball{\epsilon}{\sigma})=1$ for any $d\in\nn$ and $\epsilon>1-\frac{1}{d}$, where $\sigma$ is the maximally mixed state $\frac{1}{d}\idop$.

\subsection{Proof of Eq.~\eqref{ineq:volume_of_exeball1}}
By defining $\phi:=\arg\min_{\phi\in\puredop{\cd}}\trdist{\phi-\rho}$, we obtain $\eball{\epsilon}{\rho}\subseteq\eball{2\epsilon}{\phi}$.  For $\trdist{\psi-\phi}\leq\trdist{\psi-\rho}+\trdist{\phi-\rho}\leq2\trdist{\psi-\rho}<2\epsilon$ for any pure state $\psi\in\eball{\epsilon}{\rho}$. This completes the proof since $\mu(\eball{\epsilon}{\rho})\leq\mu(\eball{2\epsilon}{\phi})=(2\epsilon)^{2(d-1)}$ for any $\epsilon\in\left(0,\frac{1}{2}\right]$.

\subsection{Proof of Eq.~\eqref{ineq:volume_of_exeball}}
Let $\rho=\sum_{i=0}^{d-1}p_i\ketbra{i}$, where $\{\ket{i}\}_i$ is a set of eigenvectors of $\rho$ and eigenvalues are arranged in decreasing order, i.e., $p_0\geq p_1\geq \cdots$. Since $\mu(\eball{\epsilon}{\rho})$ depends not on the eigenvectors but on the eigenvalues of $\rho$, it is sufficient to consider only diagonal $\rho$ with respect to a fixed basis. However, it is difficult to exactly calculate $\mu(\eball{\epsilon}{\rho})$ due to a complicated relationship between $\psi$ and the largest eigenvalue of $\psi-\rho$, resulting from the condition $\epsilon>\trdist{\psi-\rho}=\lambda_{\max}(\psi-\rho)$. 

We derive lower bound $f_\rho(\psi)$ of $\trdist{\psi-\rho}$ and use the relationship $\mu(\eball{\epsilon}{\rho})\leq\mu\left(\{\psi:f_\rho(\psi)<\epsilon\}\right)$ to show Eq.~\eqref{ineq:volume_of_exeball}, where $f_\rho$ is a measurable function. Since simple bound $f_\rho(\psi)=1-\fidelity{\psi}{\rho}$ is too loose to show Eq.~\eqref{ineq:volume_of_exeball}, we derive a tighter lower bound as follows: 
Let $\Pi$ and $\Pi^\bot$ be the Hermitian projectors on two-dimensional subspace $\vv\supseteq\vspan{\{\ket{0},\ket{\psi}\}}$ and its orthogonal complement, respectively. We then obtain
\begin{eqnarray}
\trdist{\psi-\rho}&\geq&\trdist{\Pi(\psi-\rho)\Pi+\Pi_\bot(\psi-\rho)\Pi_\bot}\nonumber\\
\label{eq:fpsi}
&=&\trdist{\psi-\Pi\rho\Pi}+\trdist{\Pi_\bot\rho\Pi_\bot},
\end{eqnarray}
where we use the monotonicity of the trace distance under a CPTP mapping in the first inequality. Define $f_\rho(\psi)$ as the value in Eq.~\eqref{eq:fpsi}, which can be explicitly written as
 \begin{equation}
 f_\rho(\psi)= \frac{1}{2}\sqrt{(1+p_0-q)^2-4(p_0-q)|\braket{0}{\psi}|^2}+\frac{1}{2}(1-p_0-q),
\end{equation}
where $q=\bra{0_\bot}\rho\ket{0_\bot}$ and $\{\ket{0},\ket{0_\bot}\}$ is an orthonormal basis of $\vv$. The explicit formula implies $f_\rho$ is  uniquely defined (although neither $\vv$ nor $q$ is uniquely defined if $\psi=\ketbra{0}$) and continuous, thus measurable. Since $\mu(\eball{\epsilon}{\rho})=0$, satisfying Eq.~\eqref{ineq:volume_of_exeball}, if $p_0\leq1-\epsilon$, we consider the case $p_0>1-\epsilon$. By further assuming $\epsilon\in\left(0,\frac{1}{2}\right]$, we obtain
\begin{equation}
 f_\rho(\psi)< \epsilon\Leftrightarrow
 |\braket{0}{\psi}|^2> \frac{(\epsilon+p_0)(1-q-\epsilon)}{p_0-q},
\end{equation}
where the this condition is not trivial, i.e., $\frac{(\epsilon+p_0)(1-q-\epsilon)}{p_0-q}\in(0,1)$.
Assuming $d\geq4$, we calculate an upper bound on $\mu(\eball{\epsilon}{\rho})$ as follows: 
By defining $\unitary{\hh}$ as the set of unitary operators on $\hh$ and $C:\unitary{\cdim{d-1}}\rightarrow\unitary{\cd}$ as $C(U):=\ketbra{0}\oplus U$, we can show that for any unitary operator $U\in\unitary{\cdim{d-1}}$,
\begin{eqnarray}
 &&\mu\left(\{\psi\in\puredop{\cd}:f_\rho(\psi)<\epsilon\}\right)\nonumber\\
&&=\mu\left(\{\psi:f_\rho(C(U)\psi C(U)^\dag)<\epsilon\}\right)\nonumber\\
\label{eq:uinv}
&&= \int_{\puredop{\cd}}d\mu(\psi)\indicator{|\braket{0}{\psi}|^2>\frac{(\epsilon+p_0)(1-q_U-\epsilon)}{p_0-q_U}}
\end{eqnarray}
where we use the unitarily invariance of $\mu$ in the first equality, $q_U=\bra{0_\bot} C(U)^\dag\rho C(U)\ket{0_\bot}$, and $\indicator{X}\in\{0,1\}$ is the indicator function, i.e., $\indicator{X}=1$ iff $X$ is true. By integrating Eq.~\eqref{eq:uinv} with respect to the Haar measure on $\unitary{\cdim{d-1}}$ and using Fubini's theorem, we obtain
 \begin{equation}
 \mu\left(\{\psi\in\puredop{\cd}:f_\rho(\psi)<\epsilon\}\right)=\int_{\puredop{\cd}}d\mu(\psi)\int_{\puredop{\cdim{d-1}}}d\mu(\phi) \indicator{|\braket{0}{\psi}|^2>\frac{(\epsilon+p_0)(1-\fidelity{\rho}{\phi}-\epsilon)}{p_0-\fidelity{\rho}{\phi}}},
\end{equation}
where $\phi\in\puredop{\cdim{d-1}}$ is identified with a pure state on $\puredop{\cd}$ acting on subspace $\vspan{\{\ket{1},\cdots,\ket{d-1}\}}$. Using Fubini's theorem again and Eq.~\eqref{eq:volume_of_eball}, we can proceed with the calculation:
\begin{eqnarray}
 &&=\int_{\puredop{\cdim{d-1}}}d\mu(\phi)\delta(\fidelity{\rho}{\phi})^{2(d-1)}\nonumber\\
 &&=\int_{\puredop{\cdim{d-1}}}d\mu(\phi)\delta\left(\sum_{i=1}^{d-1}p_i|\braket{i}{\phi}|^2\right)^{2(d-1)}\nonumber\\
 &&\leq\int_{\puredop{\cdim{d-1}}}d\mu(\phi)\delta\left((1-p_0)|\braket{1}{\phi}|^2\right)^{2(d-1)}\nonumber\\
 &&=(d-2)\int_0^1(1-x)^{d-3}\delta((1-p_0)x)^{2(d-1)}dx=:g_{d,\epsilon}(p_0),
\end{eqnarray}
where $\delta(q)=\sqrt{\frac{(\epsilon+q)(p_0+\epsilon-1)}{p_0-q}}$, we use the convexity of $\delta^{2(d-1)}$ and the unitary invariance of $\mu$ in the last inequality, and we use the probability density of $x=|\braket{1}{\phi}|^2$ derived by Eq.~\eqref{eq:volume_of_eball} in the last equality.
To confirm the calculation, we plot a comparison between $\mu(\eball{\epsilon}{\rho})$ and its upper bound $g_{d,\epsilon}(p_0)$ for a particular $\rho$ as shown in Fig.~\ref{fig:boundcomparison}, where we use the following explicit expression of $g_{4,\epsilon}(p_0)$:
\begin{eqnarray}
 \label{eq:g4expression}
 g_{4,\epsilon}(p_0)&=&2(p_0+\epsilon-1)^3\bigg\{\frac{1-6b-ab^2}{2ab^2}+\frac{3(a+1)}{a(b-a)}\left(1-\frac{a}{b-a}\log\frac{b}{a}\right)\bigg\},
\end{eqnarray}
where $a=\frac{2p_0-1}{\epsilon+p_0}$ and $b=\frac{p_0}{\epsilon+p_0}$ and $p_0\in(1-\epsilon,1)$. Note that $g_{4,\epsilon}(1)=\epsilon^6=\lim_{p_0\rightarrow1}g_{4,\epsilon}(p_0)$.

\begin{figure}[h]
\includegraphics[height=.23\textheight]{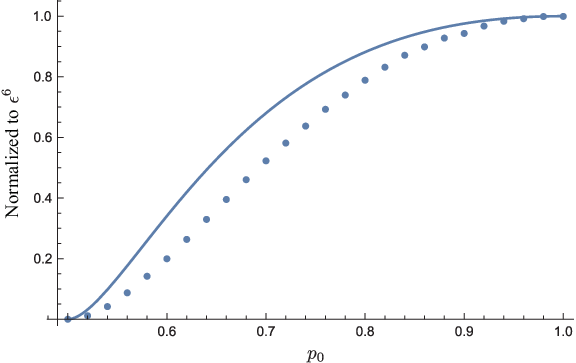}
\caption{\label{fig:boundcomparison} Plots of estimated values of $\mu(\eball{\frac{1}{2}}{\rho})$ (dots) and $g_{4,\frac{1}{2}}(p_0)$ (curve) for $\rho=p_0\ketbra{0}+(1-p_0)\ketbra{1}\in\dop{\cdim{4}}$. $\mu(\eball{\frac{1}{2}}{\rho})$ is estimated by uniformly sampling $10^7$ pure states. The plots indicate $\mu(\eball{\frac{1}{2}}{\rho})$ is well upper bounded by $g_{4,\frac{1}{2}}(p_0)$.}
\end{figure}

It is sufficient to show that under the two conditions $\epsilon\in\left(0,\frac{1}{2}\right]$ and $d\geq4$,
\begin{equation}
\forall p_0\in(1-\epsilon,1),\frac{dg_{d,\epsilon}}{dp_0}\geq0
\end{equation}
since $g_{d,\epsilon}(1)=\epsilon^{2(d-1)}$. Since the integrand of $g_{d,\epsilon}$ and its partial derivative with respect to $p_0$ are continuous, we can interchange the partial differential and integral operators:
\begin{eqnarray}
 \frac{dg_{d,\epsilon}}{dp_0}&=&(d-2)\int_0^1(1-x)^{d-3}\frac{\partial}{\partial p_0}\delta((1-p_0)x)^{2(d-1)}dx\nonumber\\
 &=&\alpha_{d,\epsilon}(p_0)\int_0^1\beta_{\epsilon}(p_0,x)\gamma_{d,\epsilon}(p_0,x)dx,
\end{eqnarray}
where $\alpha_{d,\epsilon}(p_0)=(d-2)(d-1)(p_0+\epsilon-1)^{d-2}$, $\beta_{\epsilon}(p_0,x)=-(1-p_0)^2x^2+(1-\epsilon-\epsilon^2-p_0^2)x+(1-\epsilon)\epsilon$ and $\gamma_{d,\epsilon}(p_0,x)=\frac{(1-x)^{d-3}(\epsilon+(1-p_0)x)^{d-2}}{(p_0-(1-p_0)x)^d}$. Since $\alpha_{d,\epsilon}$ and $\gamma_{d,\epsilon}$ are non-negative in the entire considered region $R:=\{(p_0,x):p_0\in(1-\epsilon,1)\wedge x\in[0,1]\}$, $\frac{dg_{d,\epsilon}}{dp_0}\geq0$ if $\beta_{\epsilon}$ is non-negative for all $x\in[0,1]$. However, $\beta_{\epsilon}$ can be negative for some $x\in[0,1]$ if and only if $\beta_{\epsilon}(p_0,1)<0$. Taking account of considered region $R$, it is sufficient to show $\frac{dg_{d,\epsilon}}{dp_0}\geq0$ for all $p_0\in\left(\frac{1+\sqrt{1-4\epsilon^2}}{2},1\right)(\subseteq(1-\epsilon,1))$, where $\beta_{\epsilon}$ can be negative. 

For fixed $p^*\in\left(\frac{1+\sqrt{1-4\epsilon^2}}{2},1\right)$, let $x^*\in(0,1)$ satisfy $\beta_{\epsilon}(p^*,x^*)=0$. Since $\beta_{\epsilon}(p^*,x)$ is monotonically decreasing in $x\geq0$, $x^*$ is uniquely defined, $\beta_{\epsilon}(p^*,x)>0$ if $x\in[0,x^*)$ and $\beta_{\epsilon}(p^*,x)<0$ if $x\in(x^*,1]$. Thus, showing
\begin{equation}
\label{eq:condition1}
 \forall d\geq4,\exists c>0,
 \left\{
\begin{array}{ll}
 \gamma_{d+1,\epsilon}(p^*,x)\geq c\gamma_{d,\epsilon}(p^*,x)&\text{for}\ x\in[0,x^*)\\
  \gamma_{d+1,\epsilon}(p^*,x)\leq c\gamma_{d,\epsilon}(p^*,x)&\text{for}\ x\in(x^*,1]
\end{array}
 \right.
\end{equation}
and
\begin{equation}
\label{eq:condition2}
\frac{dg_{4,\epsilon}}{dp_0}\bigg|_{p_0=p^*}\geq0,
\end{equation}
is sufficient for $\forall d\geq4,\frac{dg_{d,\epsilon}}{dp_0}\Big|_{p_0=p^*}\geq0$. For
\begin{eqnarray}
 \alpha_{d+1,\epsilon}(p^*)^{-1}\frac{dg_{d+1,\epsilon}}{dp_0}\bigg|_{p_0=p^*}&=&\int_0^{x^*}\beta_{\epsilon}(p^*,x)\gamma_{d+1,\epsilon}(p^*,x)dx+\int_{x^*}^{1}\beta_{\epsilon}(p^*,x)\gamma_{d+1,\epsilon}(p^*,x)dx\nonumber\\
 &\geq& c\biggl\{\int_0^{x^*}\beta_{\epsilon}(p^*,x)\gamma_{d,\epsilon}(p^*,x)dx+\int_{x^*}^{1}\beta_{\epsilon}(p^*,x)\gamma_{d,\epsilon}(p^*,x)dx\biggr\}\nonumber\\
 &=&c\alpha_{d,\epsilon}(p^*)^{-1}\frac{dg_{d,\epsilon}}{dp_0}\bigg|_{p_0=p^*}
\end{eqnarray}
holds for any $d\geq4.$

First, we show Eq.~\eqref{eq:condition1}. By observing that for any $d\geq4$,
\begin{equation}
\gamma_{d+1,\epsilon}(p^*,x)-c \gamma_{d,\epsilon}(p^*,x)= \gamma_{d,\epsilon}(p^*,x)\left(\frac{(1-x)(\epsilon+(1-p^*)x)}{p^*-(1-p^*)x}-c\right),
\end{equation}
 $h_{\epsilon,p^*}(x):=\frac{(1-x)(\epsilon+(1-p^*)x)}{p^*-(1-p^*)x}$ is monotonically decreasing in $x\in[\hat{x},1]$ and $h_{\epsilon,p^*}(x)\geq h_{\epsilon,p^*}(\hat{x})$ for $x\in[0,\hat{x}]$ with $\hat{x}:=\max\left\{0,1-\frac{2p^*-1}{p^*(1-p^*)}\epsilon\right\}(\leq x^*)$, setting $c=h_{\epsilon,p^*}(x^*)(>0)$ implies Eq.~\eqref{eq:condition1}.

Next, Eq.~\eqref{eq:condition2} can be verified by using the explicit expression Eq.~\eqref{eq:g4expression}.

\section{Proof for Thoerem \ref{thm:unitary} and its sharp lower bound}
\label{appendix:unitarymixing}
In this section, we prove Theorem~\ref{thm:unitary} with a tighter lower bound as the following theorem. After the proof, we show that this lower bound is sharp by constructing set $\{\Upsilon_x\}_x$ of unitary transformations such that $\max_{\Upsilon}\min_{p}\frac{1}{2}\diamondnorm{\Upsilon-\sum_{x\in X}p(x)\Upsilon_x}$ is arbitrarily close to its lower bound $\frac{4\beta}{d}\left(1-\frac{\beta}{d}\right)=\frac{4}{d}\left(1-\frac{1}{d}\right)$ while the upper bound given in the theorem is $\alpha=1$ (actually, each $\Upsilon_x$ in the set can be perfectly distinguished from the identity transformation,) where $\alpha$ and $\beta$ are defined in the following theorem.

\begin{theorem}
\label{thm:unitarymixing}
For an integer $d\geq2$ specified below, let $\{\Upsilon_x\}_{x\in X}$ be a finite set of unitary transformations on $\linop{\cd}$. Then, it holds that
\begin{eqnarray}
\frac{4\beta}{d}\left(1-\frac{\beta}{d}\right)\leq \max_{\Upsilon}\min_{p}\frac{1}{2}\diamondnorm{\Upsilon-\sum_{x\in X}p(x)\Upsilon_x}\leq \alpha\\
 \alpha= \max_{\Upsilon}\min_{x\in X}\left(\frac{1}{2}\diamondnorm{\Upsilon-\Upsilon_x}\right)^2, \ \ \beta=1-\sqrt{1-\alpha},
 \end{eqnarray}
where the first minimum is taken over probability distribution $p$ over $X$. 
Note that if $d=2$, the equalities hold.
\end{theorem}

The lower bound of Theorem~\ref{thm:unitary} can be derived from this theorem as follows:
\begin{equation}
\frac{4\beta}{d}\left(1-\frac{\beta}{d}\right)\geq \frac{4\beta}{d}\left(1-\frac{\beta}{2}\right)=\frac{2}{d}\alpha.
\end{equation}

In the proof, we use the following fact about the diamond norm: For two CPTP linear mappings $\Gamma$ and $\Lambda$ from $\linop{\hh_1}$ to $\linop{\hh_2}$, we can verify
\begin{equation}
\label{eq:diamondnorm}
 \frac{1}{2} \diamondnorm{\Gamma-\Lambda}=\max_{M\in\mathbf{T}(\hh_1:\hh_2)}\tr{(J(\Gamma-\Lambda))M},
\end{equation}
where $J(\Xi):=\sum_{i,j}\ket{i}\bra{j}\otimes\Xi(\ket{i}\bra{j})\in\linop{\hh_1\otimes\hh_2}$ is the Choi-Jamio\l kowski operator of linear mapping $\Xi$ and $\mathbf{T}(\hh_1:\hh_2):=\{M\in\pos{\hh_1\otimes\hh_2}:\exists\rho\in\dop{\hh_1},M\leq\rho\otimes\idop\}$ is the set of measuring strategies \cite{GJ07} or that of quantum testers \cite{GGP09}.

\begin{proof}
Let $\Upsilon$ and $\Upsilon_x$ be unitary transformations from $\linop{\hh_1}$ to $\linop{\hh_2}$, defined as $\Upsilon(\rho)=U\rho U^\dag$ and $\Upsilon_x(\rho)=U_x\rho U_x^\dag$, respectively. 

First, we show 
\begin{equation}
\label{ineq:upperdiamond}
 \max_{\Upsilon}\min_{p}\frac{1}{2}\diamondnorm{\Upsilon-\sum_{x\in X}p(x)\Upsilon_x}\leq  \max_{\Upsilon}\min_{x\in X}\left(\frac{1}{2}\diamondnorm{\Upsilon-\Upsilon_x}\right)^2.
\end{equation}

By using Eq.~\eqref{eq:diamondnorm}, we obtain
\begin{eqnarray}
  &&(\text{L.H.S.\ of\ Ineq.}\eqref{ineq:upperdiamond})\nonumber\\
  &&=\max_{\Upsilon}\min_{p}\max_{M\in\mathbf{T}(\hh_1:\hh_2)}\tr{J\left(\Upsilon-\sum_{x\in X}p(x)\Upsilon_x\right)M}\nonumber\\
  &&=\max_{\Upsilon}\max_{M\in\mathbf{T}(\hh_1:\hh_2)}\min_{p}\tr{J\left(\Upsilon-\sum_{x\in X}p(x)\Upsilon_x\right)M}\nonumber\\
  &&=\max_{M\in\mathbf{T}(\hh_1:\hh_2)}\left\{\max_{\Upsilon}\tr{J\left(\Upsilon\right)M}-\max_{x\in X}\tr{J\left(\Upsilon_x\right)M}\right\},
\end{eqnarray}
where we use the minimax theorem in the second equation since $f(p,M):=\tr{J\left(\Upsilon-\sum_xp(x)\Upsilon_x\right)M}$ is affine with respect to each variable and the domain of $p$ and $M$ are compact and convex. 
Since it is known that the set of mappings $\Upsilon\mapsto\tr{J\left(\Upsilon\right)M}$ associated to quantum testers $M$ is equivalent to that of mappings $\Upsilon\mapsto\tr{\Upsilon\otimes id_{\hh_3}(\Phi)\Pi}$ associated to pure states $\Phi$ and hermitian projectors $\Pi$ \cite[Theorem~10]{GGP09} for sufficiently large dimensional Hilbert space $\hh_3$ (to be self-contained, we provide a proof for the equivalence between the two mappings in Appendix \ref{appendix:tester}) we can proceed as follows:
\begin{eqnarray}
\label{eq:max1}
  =\max_{\substack{\Phi\in\puredop{\hh_1\otimes\hh_3}\\\Pi\in\mathbf{Proj}(\hh_2\otimes\hh_3)}}
  \Big(\max_U\tr{(U\otimes\idop_{\hh_3})\Phi (U\otimes\idop_{\hh_3})^\dag\Pi}
  -\max_{x\in X}\tr{(U_x\otimes\idop_{\hh_3})\Phi (U_x\otimes\idop_{\hh_3})^\dag\Pi}\Big),
\end{eqnarray}
where $\mathbf{Proj}(\hh)$ is the set of hermitian projectors on $\hh$. Let $\hat{\Phi}$, $\hat{\Pi}$ and $\hat{U}$ maximize Eq.~\eqref{eq:max1}.
Since arbitrary unitary transformations cannot be represented by ensembles of finite set of unitary transformations, the first term cannot be $0$, thus $\hat{\Pi}\hat{U}\ket{\hat{\Phi}}\neq0$. Let $\hat{\Psi}$ be the pure state such that $\ket{\hat{\Psi}}\propto\hat{\Pi}\hat{U}\ket{\hat{\Phi}}$. Then, we can verify that Eq.~\eqref{eq:max1} is still maximized even if we replace $\hat{\Pi}$ by $\hat{\Psi}$.
Thus, $\Pi$ in Eq.~\eqref{eq:max1} can be restricted as a pure state, i.e., $\Pi=\Psi\in\puredop{\hh_2\otimes\hh_3}$, and we proceed as follows:
\begin{eqnarray}
  =\max_{\substack{\Phi\in\puredop{\hh_1\otimes\hh_3}\\\Psi\in\puredop{\hh_2\otimes\hh_3}}}\Big(\max_U|\bra{\Psi}U\otimes\idop_{\hh_3}\ket{\Phi}|^2 -\max_{x\in X}|\bra{\Psi}U_x\otimes\idop_{\hh_3}\ket{\Phi}|^2\Big).
  \end{eqnarray}
Before proceeding to the next step, we show the set of mappings $f_{\Phi,\Psi}:U\mapsto |\bra{\Psi}U\otimes\idop_{\hh_3}\ket{\Phi}|$ associated to pure states $\Phi$ and $\Psi$ is equivalent to that of mappings $g_A:U\mapsto\left|\tr{A U}\right|$ associated to linear operator $A$ such that $\lpnorm{1}{A}\leq1$, where $\lpnorm{1}{A}$ is the Schatten $1$-norm of $A$. By using decompositions $\ket{\Phi}=\sum_{i,j}\alpha_{ij}\ket{i}\ket{j}$ and $\ket{\Psi}=\sum_{i,j}\beta_{ij}\ket{i}\ket{j}$ with respect to orthonormal bases, we can verify that $g_A$ with $A=\sum_{i,j,k}\alpha_{ik}\beta^*_{jk}\ket{i}\bra{j}$ equals to $f_{\Phi,\Psi}$ and $\lpnorm{1}{A}=\max_Ug_A(U)=\max_Uf_{\Phi,\Psi}(U)\leq1$. On the other hand, By using the singular value decomposition $A=\sum_ip_i\ket{x_i}\bra{y_i}$, where $\lpnorm{1}{A}\leq1$ implies $p+\sum_ip_i=1$ with some $p\geq0$, we can verify that $f_{\Phi,\Psi}$ with $\ket{\Phi}=\sqrt{p}\ket{0}\ket{\bot}+\sum_i\sqrt{p_i}\ket{x_i}\ket{i}$ and $\ket{\Psi}=\sqrt{p}\ket{0}\ket{\bot'}+\sum_i\sqrt{p_i}\ket{y_i}\ket{i}$ ($\{\ket{i}\}_i\cup\{\ket{\bot},\ket{\bot'}\}$ is an orthonormal basis) equals to $g_A$.
  
 By using the equivalent between two sets of mappings and $\lpnorm{1}{A}=\max_U\left|\tr{A U}\right|$, we proceed as follows:
\begin{eqnarray}
&&=\max_{A:\lpnorm{1}{A}\leq1}\left(\lpnorm{1}{A}^2-\max_{x\in X}|\tr{A U_x}|^2\right)\nonumber\\
&&=\max_{\substack{V,\rho\in\dop{\hh_1}\\q\in[0,1]}}q^2\left(1-\max_{x\in X}|\tr{\rho V^\dag U_x}|^2\right)\nonumber\\
\label{eq:mixdiamond}
&&=1-\min_{\substack{V,\rho\in\dop{\hh_1}}}\max_{x\in X}|\tr{\rho V^\dag U_x}|^2,
\end{eqnarray}
where we use the polar decomposition $A=q\rho V^\dag$ in the second equation and $V:\hh_1\rightarrow\hh_2$ is a unitary operator. 
On the other hand, since $\frac{1}{2}\diamondnorm{\Upsilon-\Upsilon_x}=\max_{\Phi\in\puredop{\hh_1\otimes\hh_1'}}\trdist{\Upsilon\otimes id_{\hh_1'}(\Phi)-\Upsilon_x\otimes id_{\hh_1'}(\Phi)}=\max_{\Phi\in\puredop{\hh_1\otimes\hh_1'}}\sqrt{1-\fidelity{\Upsilon\otimes id_{\hh_1'}(\Phi)}{\Upsilon_x\otimes id_{\hh_1'}(\Phi)}}=\sqrt{1-\min_{\Phi\in\puredop{\hh_1\otimes\hh_1'}}|\bra{\Phi}U^\dag U_x\otimes\idop_{\hh_1'}\ket{\Phi}|^2}=\sqrt{1-\min_{\rho\in\dop{\hh_1}}|\tr{\rho U^\dag U_x}|^2}$, it holds that
\begin{equation}
\label{eq:eachdiamond}
(\text{R.H.S.\ of\ Ineq.}\eqref{ineq:upperdiamond})=1-\min_V\max_{x\in X}\min_{\rho\in\dop{\hh_1}}|\tr{\rho V^\dag U_x}|^2.
\end{equation}
Since $\max_x\min_yf(x,y)\leq\min_y\max_xf(x,y)$ for any $f$ if the maximum and the minimum exist, we obtain Ineq.~\eqref{ineq:upperdiamond}.

Next, we show
\begin{equation}
\label{ineq:lowerdiamond}
 \max_{\Upsilon}\min_{p}\frac{1}{2}\diamondnorm{\Upsilon-\sum_{x\in X}p(x)\Upsilon_x}\geq  \frac{4\beta}{d}\left(1-\frac{\beta}{d}\right),
\end{equation}
where $\beta=1-\sqrt{1-\max_{\Upsilon}\min_{x\in X}\left(\frac{1}{2}\diamondnorm{\Upsilon-\Upsilon_x}\right)^2}$. Due to Eq.~\eqref{eq:eachdiamond}, we can verify that $\beta=1-\min_V\max_{x}\min_{\rho}\left|\tr{\rho V^\dag U_x}\right|$.
First, observe that for any unitary operator $W$ on $\cd$ $(d\geq2)$,
\begin{eqnarray}
 \frac{1}{d}\left|\tr{W}\right|=\frac{1}{d}\left|\sum_{i=1}^d\lambda_i(W)\right|
 \leq\frac{2}{d}\min_p\left|\sum_{i=1}^dp(i)\lambda_i(W)\right|+\frac{d-2}{d}
=\frac{2}{d}\min_{\rho\in\dop{\cd}}\left|\tr{\rho W}\right|+\frac{d-2}{d}
\end{eqnarray}
holds, where $\lambda_i(W)$ is the $i$-th eigenvalue of $W$, and in the inequality, we use the following two facts: (i) the minimization is achieved only if $p$ satisfies $\forall i,p(i)\leq\frac{1}{2}$ due to a geometric obseravation, and (ii) for such probability distribution $p\left(\leq\frac{1}{2}\right)$ and complex numbers $\lambda_i\in\{z\in\mathbb{C}:|z|=1\}$, $\left|\sum_ip(i)\lambda_i\right|\geq\left|\sum_i\frac{1}{2}\lambda_i\right|-\left|\sum_i\left(\frac{1}{2}-p(i)\right)\lambda_i\right|\geq\frac{1}{2}\left|\sum_i\lambda_i\right|-\sum_i\left(\frac{1}{2}-p(i)\right)=\frac{1}{2}\left|\sum_i\lambda_i\right|-\frac{d-2}{2}$. By using this, we obtain
\begin{eqnarray}
&& \min_{\substack{V,\rho\in\dop{\hh_1}}}\max_{x\in X}|\tr{\rho V^\dag U_x}|\leq\min_{\substack{V}}\max_{x\in X}\left|\tr{  \frac{\idop}{d}V^\dag U_x}\right|\nonumber\\
&&\leq\frac{2}{d}\min_{\substack{V}}\max_{x\in X}\min_{\rho\in\dop{\cd}}\left|\tr{\rho V^\dag U_x}\right|+\frac{d-2}{d}
=\frac{2}{d}(1-\beta)+\frac{d-2}{d}.
\end{eqnarray}
This and Eq.~\eqref{eq:mixdiamond} implies
\begin{eqnarray}
 (\text{L.H.S.\ of\ Ineq.}\eqref{ineq:lowerdiamond})\geq 1-\left(\frac{2}{d}(1-\beta)+\frac{d-2}{d}\right)^2
 = (\text{R.H.S.\ of\ Ineq.}\eqref{ineq:lowerdiamond}).
\end{eqnarray}
This completes the proof.

\end{proof}

\subsection{Sharpness of the lower bound}
We show the equality in Ineq.~\eqref{ineq:lowerdiamond} holds when we approximate arbitrary unitary transformations by using restricted set $\{\Lambda:\Lambda(\rho)=W\rho W^\dag,W\in\mathbf{W}\}$ of unitary transformations, where
\begin{eqnarray}
\mathbf{W}:=\left\{W:
\begin{matrix}
 W\text{\ is\ a\ unitary\ operator\ on\ $\cd$ s.t.}\\
 \text{the\ convex\ hull\ of\ }W\text{'s\ eigenvalues\ contains\ }0.
\end{matrix}
\right\}.
\end{eqnarray}
 More precisely, since we assume $\{\Upsilon_x\}_x$ is a finite set, we show that Ineq.~\eqref{ineq:lowerdiamond} is getting tighter when $\epsilon\rightarrow0$ and $\{\Upsilon_x:\Upsilon_x(\rho)=U_x\rho U_x^\dag\}_{x}$ is an $\epsilon$-covering of $\{\Lambda:\Lambda(\rho)=W\rho W^\dag,W\in\mathbf{W}\}$, i.e., $\max_\Lambda\min_x\diamondnorm{\Lambda-\Upsilon_x}<\epsilon$, where $U_x\in\mathbf{W}$.
This can be shown by the following two observations:

First, for any $x\in X$, there exists pure state $\phi_x\in\puredop{\cd}$ such that $\tr{\phi_x U_x}=0$ since $0$ is contained in the convex hull of $U_x$'s eigenvalues. Then, we obtain
\begin{eqnarray}
 &&\max_{\Upsilon}\min_{x\in X}\frac{1}{2}\diamondnorm{\Upsilon-\Upsilon_x}\geq\min_{x\in X}\frac{1}{2}\diamondnorm{id-\Upsilon_x}\geq\min_{x\in X}\trdist{\phi_x-\Upsilon_x(\phi_x)}\nonumber\\
 &&=\min_{x\in X}\sqrt{1-\fidelity{\phi_x}{U_x\phi_xU_x^\dag}}=\min_{x\in X}\sqrt{1-|\tr{\phi_xU_x}|^2}=1.
\end{eqnarray}
This implies that $\beta$ in Ineq.~\eqref{ineq:lowerdiamond} satisfies $\beta=1$.

Second, by using Eq.~\eqref{eq:mixdiamond}, we obtain
\begin{eqnarray}
 \max_{\Upsilon}\min_{p}\frac{1}{2}\diamondnorm{\Upsilon-\sum_{x\in X}p(x)\Upsilon_x}
=1-\min_{V,\rho\in\dop{\cd}}\max_{x\in X}\left|\tr{\rho V^\dag U_x}\right|^2
 < 1-\min_{V,\rho\in\dop{\cd}}\max_{W\in\mathbf{W}}\left|\tr{\rho V^\dag W}\right|^2+\epsilon,
\end{eqnarray}
where in the inequality, we use that for any $W\in\mathbf{W}$, there exists $\Upsilon_x$ such that $\epsilon>\frac{1}{2}\diamondnorm{\Lambda-\Upsilon_x}\geq\tr{\Phi(V^\dag W\otimes\idop)\Phi( W^\dag V\otimes\idop)}-\tr{\Phi(V^\dag U_x\otimes\idop)\Phi(U_x^\dag V\otimes\idop)}=\left|\tr{\rho V^\dag W}\right|^2-\left|\tr{\rho V^\dag U_x}\right|^2$, where $\Lambda(\rho)=W\rho W^\dag$ and $\Phi$ is a purification of $\rho$.
 Since $\epsilon$ can be arbitrarily small positive number, showing 
 \begin{equation}
 \label{eq:unitaryapproximation}
\min_{V,\rho\in\dop{\cd}}\max_{W\in\mathbf{W}}\left|\tr{\rho V^\dag W}\right|\geq1-\frac{2}{d}
\end{equation}
is sufficient to prove the sharpness of Ineq.~\eqref{ineq:lowerdiamond}.

For any unitary operator $V$, there exists $\{W_i\in\mathbf{W}\}_{i=1}^d$ such that $V, W_1,\cdots, W_d$ are simultaneously diagonalizable and the $j$-th eigenvalues of $V$ and $W_i$ are the same for all $j\neq i$.
By letting $V^\dag W_i=\idop-\left(1-z_i\right)\ketbra{i}$, where complex number $z_i$ satisfies $|z_i|=1$ and $\{\ket{i}\}_{i=1}^d$ is an orthonormal basis, we can proceed as follows: for any unitary operator $V$,
\begin{eqnarray}
&& \min_{\rho\in\dop{\cd}}\max_{W\in\mathbf{W}}\left|\tr{\rho V^\dag W}\right|\geq\min_{\rho\in\dop{\cd}}\max_{1\leq i\leq d}\left|\tr{\rho (\idop-\left(1-z_i\right)\ketbra{i})}\right|\nonumber\\
&&\geq\min_{\rho\in\dop{\cd}}\max_{1\leq i\leq d}\left(1-\left|1-z_i\right|\bra{i}\rho\ket{i}\right)\geq1-2\max_{\rho\in\dop{\cd}}\min_{1\leq i\leq d}\bra{i}\rho\ket{i}\geq1-\frac{2}{d}.
\end{eqnarray}
This completes the proof.

\section{The operator norm and the diamond norm of unitary transformations}
\label{appendix:diamondnorm}
Suppose $\delta\in[0,\sqrt{2}]$. In this section, we show that $\diamondnorm{\Upsilon_1-\Upsilon_2}<\delta\sqrt{4-\delta^2}$ if $\lpnorm{\infty}{U_1-U_2}<\delta$, where $\Upsilon_i(\rho):=U_i\rho U_i^\dag$ is a unitary transformation on $\linop{\hh}$.

By using the unitarily invariance of the diamond norm and the operator norm, it is sufficient to show $\diamondnorm{id-\Upsilon}<\delta\sqrt{4-\delta^2}$ if $\lpnorm{\infty}{\idop-U}<\delta$, where $\Upsilon(\rho):=U\rho U^\dag$. Let $\lambda_i\in\{z\in\mathbb{C}:|z|=1\}$ be the $i$-th eigenvalue of $U$. Then, we obtain
\begin{eqnarray}
\label{ineq:opnorm}
 \delta>\lpnorm{\infty}{\idop-U}=\max_i\{|1-\lambda_i|\}.
\end{eqnarray}
On the other hand, by using the similar observation used for deriving Eq.~\eqref{eq:eachdiamond},
\begin{eqnarray}
 \diamondnorm{id-\Upsilon}&=&2\sqrt{1-\min_{\rho\in\dop{\hh}}|\tr{\rho U}|^2}=2\sqrt{1-\min_{p}\left|\sum_ip(i)\lambda_i\right|^2}.
\end{eqnarray}
By using an elementary geometric observation shown in Fig.\ref{fig:diamondnorm}, we obtain $\min_{p}\left|\sum_ip(i)\lambda_i\right|>\epsilon=\sqrt{1-\frac{\delta^2}{4}\left(4-\delta^2\right)}$. This completes the proof.

\begin{figure}[h]
\includegraphics[height=.18\textheight]{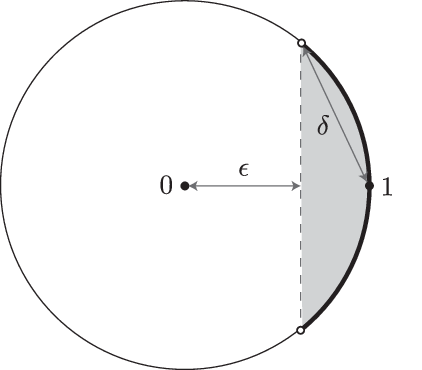}
\caption{\label{fig:diamondnorm} Geometric relationship between the operator norm and the diamond norm. Ineq.~\eqref{ineq:opnorm} implies that eigenvalues $\lambda_i$ of $U$ reside on the bold arc and their convex hull is contained in the shaded segment, which also implies that the shortest distance between the origin and the convex hull is greater than $\epsilon$.}
\end{figure}

\section{Circuit synthesis of single qubit unitary transformations by using gate set $\{S,H,T\}$}
\label{appendix:synthesis}
Recall that $n(\epsilon)$ is the smallest length of gate sequences formed from gate set $\{S,H,T\}$ to approximate arbitrary single qubit unitary transformations within accuracy $\epsilon$, i.e., $n(\epsilon):=\min\left\{n\in\nn:\max_\Upsilon\min_x\frac{1}{2}\diamondnorm{\Upsilon-\Upsilon_x^{(n)}}<\epsilon\right\}$, where $\{\Upsilon_x^{(n)}\}_{x}$ is the set of unitary transformations representing the unitary circuit realized by the gate sequences of length at most $n$.
In this section, we perform a numerical calculation of $n(\epsilon)$ and show that the probabilistic implementation reduces the gate length owing to Theorem \ref{thm:unitary}.

First, we generate $\{\Upsilon_x^{(n)}\}_{x}$. Next, we randomly sample $2\times 10^4$ single-qubit unitary transformations with respect to the Haar measure on the unitary group on $\cdim{2}$. Third, for each randomly sampled unitary transformation $\Upsilon$, we calculate the half $\hat{\epsilon}(\Upsilon)$ of the minimum diamond norm between $\Upsilon$ and $\{\Upsilon_x^{(n)}\}_{x}$. Then, we compute the maximum value $\hat{\epsilon}$ of $\hat{\epsilon}(\Upsilon)$ over all the sampled $\Upsilon$. Since another numerical experiment indicates that the set of randomly sampled $2\times 10^4$ single-qubit unitary transformations is $0.1$-covering of that of single-qubit unitary transformations with high probability, we assume that gate sequences of length at most $n$ approximate arbitrary single-qubit unitary transformations within accuracy $\hat{\epsilon}$. (The true accuracy $\epsilon$ satisfies $\hat{\epsilon}\leq\epsilon<\hat{\epsilon}+0.1$.) 

In Fig.\ref{fig:synthesis}, we plot $n(\epsilon)$ and $n(\sqrt{\epsilon})$, which correspond to the minimum length of gate sequences to synthesize single-qubit unitary transformations within accuracy $\epsilon$ by using the deterministic implementation and the probabilistic one, respectively. This graph shows that the probabilistic implementation can reduce the circuit size in a wide range of accuracy $\epsilon$.

\begin{figure}[h]
\includegraphics[height=.22\textheight]{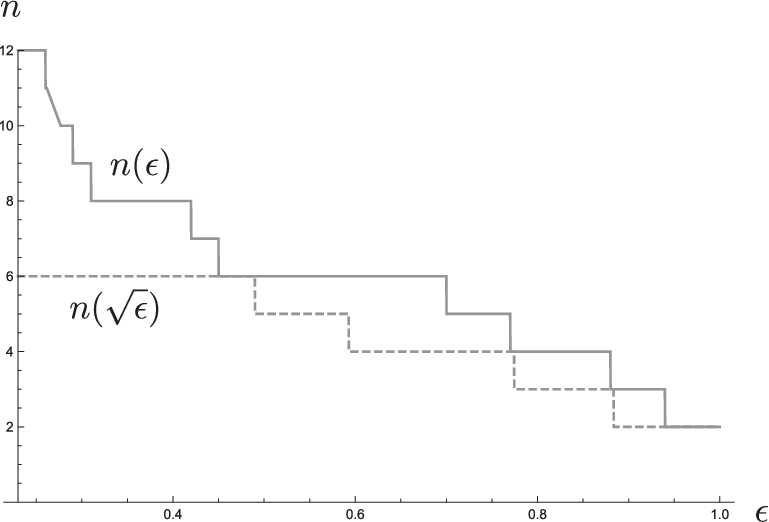}
\caption{\label{fig:synthesis} The minimum gate length at most $n$ to synthesize single qubit unitary transformations with gate set $\{S,H,T\}$ within accuracy $\epsilon$. The solid graph and the dashed graph correspond to the deterministic implementation and the probabilistic one, respectively. Note that in both cases, the accuracy is measured by using the half of the diamond norm, i.e., $\max_\Upsilon\min_x\frac{1}{2}\diamondnorm{\Upsilon-\Upsilon_x}$ and $\max_\Upsilon\min_p\frac{1}{2}\diamondnorm{\Upsilon-\sum_xp(x)\Upsilon_x}$, respectively.}
\end{figure}

\section{Equivalence between quantum testers and quantum networks}
\label{appendix:tester}
Recall that  the Choi-Jamio\l kowski operator of linear mapping $\Xi:\linop{\hh_1}\rightarrow\linop{\hh_2}$ is defined as $J(\Xi):=\sum_{i,j}\ket{i}\bra{j}\otimes\Xi(\ket{i}\bra{j})\in\linop{\hh_1\otimes\hh_2}$, and the set of quantum testers is defined as $ \mathbf{T}(\hh_1:\hh_2):=\{M\in\pos{\hh_1\otimes\hh_2}:\exists\rho\in\dop{\hh_1},M\leq\rho\otimes\idop\}$.
 In this section, we show that the set of mappings $f_M:\Xi\mapsto\tr{J\left(\Xi\right)M}$ associated to quantum testers $M\in \mathbf{T}(\hh_1:\hh_2)$ is equivalent to that of mappings $g_{\Phi,\Pi}:\Xi\mapsto\tr{\Xi\otimes id_{\hh_3}(\Phi)\Pi}$ associated to pure states $\Phi$ and hermitian projectors $\Pi$ for sufficiently large dimensional Hilbert space $\hh_3$. Note a proof for more general quantum testers is given in \cite[Theorem~10]{GGP09}.

First, we show that for any $\Phi\in\puredop{\hh_1\otimes\hh_3}$ and $\Pi\in\mathbf{Proj}(\hh_2\otimes\hh_3)$, there exists $M\in \mathbf{T}(\hh_1:\hh_2)$ such that $f_M=g_{\Phi,\Pi}$ as follows:
By letting $M=\ptr{3}{(\Phi^{T_1}\otimes\idop_2)(\idop_1\otimes\Pi)}$, we obtain
\begin{eqnarray}
g_{\Phi,\Pi}(\Xi)=\tr{\Xi\otimes id_{\hh_3}(\Phi)\Pi}&=&\tr{(J(\Xi)\otimes\idop_3)(\Phi^{T_1}\otimes\idop_2)(\idop_1\otimes\Pi)}=\tr{J(\Xi)M}=f_M(\Xi),
\end{eqnarray}
where $\Phi^{T_1}$ and $\ptr{3}{\cdot}$ represent the partial transpose of $\Phi$ and the partial trace, and the subscript of the operator denotes the system where the operator acts on. We can also verify that $M\in \mathbf{T}(\hh_1:\hh_2)$ as follows: Let $X=\sum_{ij}\alpha_{ij}\ket{j}_3\bra{i}_1$, where $\ket{\Phi}=\sum_{ij}\alpha_{ij}\ket{i}_1\ket{j}_3$ with the computational basis $\{\ket{i}_1\in\puredop{\hh_1}\}_i$ and $\{\ket{j}_3\in\puredop{\hh_3}\}_j$. Then, we obtain that for any postive semidefinte operator $P\in\pos{\hh_1\otimes\hh_2}$,
\begin{equation}
\tr{PM}=\tr{(P\otimes\idop_3)(\Phi^{T_1}\otimes\idop_2)(\idop_1\otimes\Pi)}=\tr{(X\otimes\idop_2)P(X\otimes\idop_2)^\dag\Pi}\geq0,
\end{equation}
which implies $M\geq0$. By letting $\rho=\ptr{3}{\Phi^{T_1}}=\ptr{3}{\Phi^T}(\in\dop{\hh_1})$, we can also verify that
\begin{eqnarray}
\rho\otimes\idop_2-M=\ptr{3}{(\Phi^{T_1}\otimes\idop_2)(\idop_{123}-\idop_1\otimes\Pi)}=\ptr{3}{(\Phi^{T_1}\otimes\idop_2)(\idop_1\otimes\Pi_\bot)}\geq0,
\end{eqnarray}
where $\Pi_\bot\in\mathbf{Proj}(\hh_2\otimes\hh_3)$ satisfies $\Pi+\Pi_\bot=\idop$, and the last inequality can be verified by the fact that $M\geq0$.

Next, we show that for any $M\in \mathbf{T}(\hh_1:\hh_2)$, there exist $\Phi\in\puredop{\hh_1\otimes\hh_3}$ and $\Pi\in\mathbf{Proj}(\hh_2\otimes\hh_3)$ such that $f_M=g_{\Phi,\Pi}$ as follows: Let $M\leq\rho_1\otimes\idop_2$, $\Phi\in\puredop{\hh_1\otimes\hh_{1'}}$ be a purification of $\rho_1^T$, its singular value decomposition be $\ket{\Phi}=\sum_i\sqrt{p(i)}\ket{x_i}_1\ket{y_i}_{1'}$ ($p(i)>0$) and $P\in\pos{\hh_2\otimes\hh_{1'}}$ be $P=XMX^\dag$, where $X=\sum_{i}\frac{1}{\sqrt{p(i)}}\ket{y_i}_{1'}\bra{x_i^*}_1$ and $\ket{\phi^*}$ is the complex conjugate of $\ket{\phi}$. Then we can verify that
\begin{equation}
f_M(\Xi)=\tr{J(\Xi)M}=\tr{(J(\Xi)\otimes\idop_{1'})(\Phi^{T_1}\otimes\idop_2)(\idop_1\otimes P)}=\tr{\Xi\otimes id_{\hh_{1'}}(\Phi)P}.
\end{equation}
Since $P\leq X(\rho_1\otimes\idop_2)X^\dag\leq\idop_{21'}$, $\{P,\idop-P\}$ is a positive operator-valued measure (POVM), which can be embedded in a larger Hilbert space $\hh_2\otimes\hh_3$ as a projection-valued measure (PVM) $\{\Pi,\Pi_\bot\}$ owing to the Naimark's extension. This completes the proof.

\end{document}